\documentclass{elsarticle}

\usepackage{algorithmic}

\usepackage{lmodern}
\usepackage[T1]{fontenc}
\usepackage{amssymb}
\usepackage{amsfonts}
\usepackage{amsthm}
\usepackage{amsmath}
\usepackage{pstricks,pst-node,pst-plot}
\usepackage{graphicx}
\usepackage[ngerman,english]{babel}

\usepackage{multirow}

\usepackage{framed}
\usepackage{bm}

\usepackage[toc,page]{appendix}

\usepackage{thmtools}
\usepackage{thm-restate}

\usepackage{geometry}
\begin{document}

\title{A note on Reed's Conjecture about $\omega$, $\Delta$ and $\chi$ with respect to vertices of high degree}
\author{Vera Weil}
\ead{vera.weil@oms.rwth-aachen.de}
\address{RWTH Aachen University, Management Science, Kackertstr.~7, 52072 Aachen}

\newtheorem{theorem}{Theorem}
\newtheorem{lemma}{Lemma}
\newtheorem{observation}{Observation}
\newtheorem{proposition}{Proposition}
\newtheorem{conjecture}{Conjecture}
\newtheorem*{conjectureohne}{Conjecture}
\newtheorem{corollary}{Corollary}

\begin{abstract}
Reed conjectured that for every graph, 
$\chi \leq \left \lceil \frac{\Delta + \omega + 1}{2} \right \rceil$  holds, where 
$\chi$, $\omega$ and $\Delta$ denote 
the chromatic number, clique number and maximum degree of the graph, respectively. 
We develop an algorithm which takes a hypothetical counterexample as input. 
The output discloses some hidden structures closely related to high vertex degrees.    
Consequently, we deduce two graph classes where Reed's Conjecture holds: 
One contains all graphs in which the vertices of degree at least $5$ form a stable set. 
The other contains all graphs in which every induced cycle of odd length contains a vertex of at most degree 3.
\end{abstract}

\maketitle

\section{Introduction}

One of the most prominent problems in combinatorial optimization 
is to decide 
whether the vertices of a graph can be feasibly colored by not more than a fixed number of different colors. 
If this fixed number is at least three, 
the aforementioned decision problem is known to be NP-complete. 
The associated optimization problem consists of computing the chromatic number $\chi$ of a graph. 
The determination of bounds for~$\chi$ is a commonly used method to confine this optimization problem.

A lower bound for $\chi$ is the clique number $\omega$. 
A classical upper bound for $\chi$ in terms of the maximum 
degree  $\Delta$ is provided by Brooks' Theorem (\cite{Brooks}).   
It implies the bound $\chi \leq \Delta+1$,   
which can be established due to an algorithmic approach (see~\cite{Lovasz_1975}, for example).  
Reed conjectured that, roughly speaking, 
the arithmetic medium of those bounds yields a new one for $\chi$. 

\begin{conjecture}(\citep{Reed1})\label{RC inequality}
For every graph $G$, 
$\chi(G) \leq \left \lceil \frac{\Delta(G) + \omega(G) + 1}{2} \right \rceil$  
holds. 
\end{conjecture}

Verifying this highly non-trivial conjecture 
would for instance imply that, essentially, 
the chromatic number never exceeds half the maximum degree 
in triangle-free graphs. 
One famous triangle-free graph is the Chv\'{a}tal graph 
(cf.~\cite{Chvatal}). 
It provides one of the numerous graphs that require the rounding up. 

For the following list of results on Reed's Conjecture, 
we spare giving the definitions of the graph notions and refer the reader to the cited literature.   
The hereafter mentioned results point out some highlights 
of the ongoing research on Reed's Conjecture. 
However, note that we do not claim that this list is exhaustive.
 
Conjecture~\ref{RC inequality} is introduced by Reed in~\cite{Reed1}. 
There, he demonstrates the following result. 

\begin{theorem}[Reed~\cite{Reed1}]
\label{theorem Reed result constant}
There is a constant $\Delta_0$ such that for $\Delta \geq \Delta_0$, 
if $G$ is a graph of maximum degree $\Delta$ with no clique of size exceeding $k$ for some 
$k \geq \left\lfloor(1-\frac{1}{70000000})\Delta\right\rfloor,$
then $\chi(G) \leq \frac{\Delta+1+k}{2}$. 
\end{theorem}

Randerath and Schiermeyer~\cite{RanderathSchiermeyer2006} work with 
such constants, too. 
\begin{theorem}[Randerath et al.~\cite{RanderathSchiermeyer2006}]\label{RanderathSchiermeyer2006}
For every $k \geq 3$ there is a constant $c_k$ such that 
$\chi(G) \leq \frac{\Delta(G)+\omega(G)+1}{2}$ 
for all graphs $G$ with 
$\Delta(G) \geq \frac{2|V(G)|}{k}+c_k\cdot(\omega(G))^{k-1}.$
\end{theorem}

Some authors presented results in terms of graph parameters such as the maximum degree, the clique number, the number of vertices, 
and the stability number. 
The latter indicates the size of a maximum stable set of a graph and is denoted by $\alpha$. 
For example, Reed~\citep{Reed1} proved that graphs with 
$\Delta(G)=|V(G)|-1$ comply with the conjecture. 
Randerath and Schiermeyer~\cite{RanderathSchiermeyer2006} improved this bound to $\Delta(G) \geq |V(G)|-4$. 
Kohl and Schiermeyer~\cite{Kohl_Schiermeyer_2010} 
showed that even less restrictive conditions to the maximum degree 
suffice since the conjecture holds for a graph $G$ with 
$\Delta(G) \geq |V(G)|-7$  
or
$\Delta(G) \geq |V(G)|-\alpha(G)-4$. 
The same authors 
verified the 
conjecture for 
graphs $G$ that obey both $\omega(G) \leq 2$ and 
$\Delta(G) \geq \frac{8(|V(G)|-\alpha(G))+118}{21}$. 
If a graph $G$ fulfills  
$\chi(G) > \left\lceil\frac{|V(G)|}{2}\right\rceil$ 
or $\chi(G)> \frac{|V(G)|-\alpha(G)+3}{2}$, 
then it obeys the conjecture, 
which was demonstrated by 
Rabern~\cite{Rabern1}. 
Gernert and Rabern~\cite{GernertRabern1} observed 
that the conjecture is valid if 
$\omega(G)=\Delta(G)$ or 
if $\omega(G)=\Delta(G)+1$ 
or if $\chi(G)\leq \omega(G) + 2.$ 

The result mentioned last implies that Conjecture~\ref{RC inequality} 
holds for perfect graphs.\label{perfect page} 
Many other hereditary graph classes have been 
analyzed by means of the conjecture. 
King, Reed and Vetta~\cite{King_etal_2007} proved that 
the conjecture holds for the line graph of a graph as well as for the 
line graph of a multigraph.
King and Reed~\cite{King_Reed_2008} showed that 
quasi-line graphs
also comply with the conjecture. 
This graph family properly contains the set of line graphs. 
Aravind, Karthick and Subramanian~\cite{Aravindetal} 
proved that every  
circular interval graph, and in particular,  
every complete expansion of $C_5$ 
obeys the conjecture. 
These results can be extended to the class of claw-free graphs, 
according to King~\cite{King2PhD}.  
The class of claw-free graphs in particular contains the 
$2K_1$-free graphs, 
that is, the complements of triangle-free graphs. 
Complements of triangle-free graphs 
form a subset of the class of almost-split graphs,  
which were introduced by 
Kohl and Schiermeyer~\cite{Kohl_Schiermeyer_2010}.  
They verified the conjecture for this graph class. 
Recently, Aravind, Karthick and Subramanian~\cite{Aravindetal} verified the conjecture for 
some further graph classes which are defined by families of forbidden induced subgraphs, such as 
the $\{$odd hole$\}$-free graphs. 
Assuming that a considered graph $G$ contains an odd hole (otherwise, the previously mentioned result takes effect), 
they showed that the conjecture holds for $G$ if 
$G$ is $\{P_5$, house, dart, $\overline{P_2 \cup P_3}\}$-free, or if 
$G$ is $\{P_5$, kite, bull, $(K_3 \cup K_1) \oplus K_1\}$-free, or if 
$G$ is $\{P_5, C_4\}$-free, or if 
$G$ is $\{$chair, bull, house, $W_4\}$-free, or if 
$G$ is $\{$chair, bull, house, dart$\}$-free.
Finally, Reed's Conjecture holds for planar and for toroidal graphs. 
This was deduced from some results of 
Thomassen~\cite{Thomassen1994}
as well as results of Albertson and Hutchinson~\cite{AlbertsonHutchinson1980} by 
Gernert and Rabern~\cite{GernertRabern2}.    

Finally, Rabern~\cite{Rabern1} observed that    
the join of two vertex-disjoint graphs obey the conjecture.  
In other words, 
Reed's Conjecture holds for graphs with 
disconnected complement. 
Therefore, a connected complement is required for every - yet hypothetical - counterexample to Conjecture~\ref{RC inequality}.  
Further properties of such a counterexample, 
mostly given in terms of graph parameters such as 
vertex and edge number, 
are listed by Gernert and Rabern~\cite{GernertRabern1}. 
Their way to attack the conjecture led to our approach:  
we assume the existence of a counterexample which is then used as an input to an algorithm. 
The output of the algorithm discloses some immanent structures of the counterexample which are formed by high degree vertices. 
As a result, we deduce two new, non-trivial graph classes that obey Reed's Conjecture.

\section{Preparation}

In order to present the main results correctly and compactly, 
we need some preparation. 
A graph class is called \textit{hereditary} 
if it is closed under taking induced subgraphs. 
Thus, a hereditary graph class allows a description in terms of 
\textit{forbidden induced subgraphs}. 
If $S$ is a set of forbidden induced subgraphs for a graph class $C$, 
then every graph in $C$ is $S$-\textit{free}. 

All graphs considered in this paper are finite, undirected and free of loops or parallel edges. 
Let $G$ denote such a graph. 
We say that a \textit{vertex coloring} is feasible for $G$ if 
a color is assigned to every vertex of $G$ such that 
adjacent vertices get different colors. 
The least number of colors needed for such a feasible vertex coloring 
is called \textit{chromatic number} and denoted by $\chi$.  
A \textit{clique} is a complete graph. 
The size of a largest clique in $G$ 
is called \textit{clique number} of $G$ and 
is denoted by $\omega$. 
A set of vertices in $G$ 
is called \textit{stable} if the vertices in the set 
are pairwise not adjacent. 

For a vertex $v$ in a graph $G$, 
a vertex adjacent to $v$ is called a \textit{neighbor}. 
The number of neighbors of $v$ in $G$ 
is the $G$-\textit{degree} of $v$, sometimes denoted by $deg_G(v)$. 
If it is clear from the context which graph is considered, 
we simply write \textit{degree} instead of $G$-degree. 
The minimum of all vertex degrees in $G$ is called 
the \textit{minimum degree} of $G$ and is denoted by $\delta(G)$; 
the maximum of all vertex degrees in $G$ is called the 
\textit{maximum degree} of $G$ and is denoted by $\Delta(G)$.   
If the $G$-degree 
of both endvertices of an edge is $\Delta(G)$, 
then the edge is a \textit{heavy edge}.
By $C_n$, $n \in \mathbb{N}$, $n \geq 3$, 
we denote a cycle on $n$ vertices. 
A cycle $C_n$ is called \textit{odd} 
if $n$ is odd. 
Moreover, 
a cycle $C$ is a \textit{heavy cycle} of $G$ 
if for all $v \in V(C)$, 
the $G$-degree of $v$ is at least $\Delta(G)-1$. 
By $K_n$, $n \in \mathbb{N}$, we denote the complete graph on $n$ vertices.

A graph $G$ is a \textit{minimal counterexample to Reed's Conjecture} 
if $$\left\lceil\frac{\Delta(G)+\omega(G)+1}{2} \right\rceil < \chi(G)$$
holds and all proper induced subgraphs of $G$ obey the conjecture.   
That is, a minimal counterexample is an inclusionwise vertex-minimal counterexample to 
Reed's Conjecture. 
Since there is no risk of confusion, 
in the \textit{proofs} of our claims we 
abbreviate \textit{counterexample to Reed's Conjecture} by \textit{counterexample}.

A vertex $v \in V(G)$ is called \textit{color-critical}  
if $\chi(G-v)<\chi(G)$. 
A graph $G$ is called \textit{color-critical}  
if every vertex in $G$ is color-critical. 
If, in addition, $\chi(G)=k$, 
then $G$ is sometimes called $k$\textit{-color-critical}. 

\begin{proposition}[Toft~\cite{Toft1974}, Diestel~\cite{DiestelBookGT}]
\label{theorem proposition k vertex critical subgraphs}
A graph $G$ contains at least one induced $k$-color-critical subgraph $H_k$ for 
$1 \leq k \leq \chi(G)$ with $\chi(H_k)=k$.
\end{proposition}

This result follows directly from Theorem 2.1 by 
Toft~\cite{Toft1974} for the case $k < \chi(G)$. 
The main idea of the proof for $k = \chi(G)$ is to 
successively remove vertices from a $k$-chromatic graph 
until the graph becomes critically 
k-chromatic, as suggested by Diestel~\cite{DiestelBookGT}.
We mimic this idea of successively deleting vertices 
in the algorithm HEAVY STABLE SETS presented below. 

In order to prove our claims, 
we need some results on color-critical graphs. 
In~\cite{Toft1974}, Toft stated that 
the only  
1-color-critical graph is 
$K_1$ and that the only $2$-color-critical graph is $K_2$. 
In addition, 
Toft recapitulated a result of 
K\"onig (referring to~\cite{Koenig1950}): 
the only $3$-color-critical graphs 
are the odd cycles. 
To facilitate citation, 
we wrap them into the following proposition. 
\begin{proposition}[Toft~\cite{Toft1974}, K\"onig~\cite{Koenig1950}]
\label{theorem proposition Toft Koenig1950}
 Let $k \in \mathbb{N}$ and let $G$ be a $k$-color-critical graph. 
 Then 
 \begin{itemize}
    \item $G \cong K_1$ if and only if $k=1$,
  \item $G \cong K_2$ if and only if $k=2$,
  \item $G \cong C_{2n+1}$, $n \in \mathbb{N}$, if and only if $k=3$.
 
 \end{itemize}
\end{proposition}

In color-critical graphs, 
every vertex has at least $\chi(G)-1$ neighbors. 

\begin{proposition}\label{theorem proposition min degree vertex critical}
Let $G$ be a color-critical graph.
Then $\delta(G) \geq \chi(G)-1$.  
\end{proposition}

Observe that Reed's Conjecture holds for graphs having 
maximum degree at most 4. 

\begin{proposition}[Gernert et al.~\cite{GernertRabern2}]
\label{theorem proposition CE has Delta at least 5}
Let $G$ be a counterexample to Reed's Conjecture. 
Then $\Delta(G) \geq 5$.
\end{proposition}

\section{Main results}

In the following, 
we present the algorithm \textsc{heavy stable sets} (cf.~Figure~\ref{figure procedure heavy}).  
\begin{figure}
\label{procedure heavys}
\begin{algorithmic}
\REQUIRE Graph $G$, $k \in \{1,\ldots,\chi(G)-1\}$
\ENSURE For all $i\in \{0,\ldots,2k\}$: Stable Sets $S_i$, Graphs $G_{i+1}$ with $G_{i+1}=G_{i}-S_{i}$

\STATE $G_0 \leftarrow G$ 
\STATE $r \leftarrow 0$
\FORALL{$i=0$ \TO $2k-1$}
\STATE $S_i \leftarrow \emptyset$
\ENDFOR
\WHILE{$\Delta(G_r-S_r)=\Delta(G)-r$}
\STATE choose $v \in V(G_r-S_r)$ such that the $(G_r-S_r)$-degree of $v$ is $\Delta(G)-r$, break ties by selecting a vertex with higher $G$-degree
\STATE $S_r \leftarrow S_r \cup \{v\}$
\STATE $G_{r+1} \leftarrow G_r-S_r$
\WHILE{$\Delta(G_r-S_r) < \Delta(G)-r$}
\IF{$r=2k-1$}
\RETURN
\ELSE 
\STATE $r \leftarrow r+1$
\STATE $G_r \leftarrow G_{r-1}-S_{r-1}$
\ENDIF
\ENDWHILE
\ENDWHILE
\end{algorithmic}
\caption{The procedure HEAVY STABLE SETS}
\label{figure procedure heavy}
\end{figure}
Applied to a given graph $G$, 
this algorithm removes one by one vertices. 
First, 
it chooses a vertex of maximum degree, 
and removes it from $G$, 
hence constructing $G'$. 
In this new graph, 
the second vertex is chosen such that it is of maximum degree 
in $G'$. 
This idea is iterated in every step. 
In other words, 
in the graph that results by removing all vertices that were already chosen, 
the algorithm chooses a vertex which is of maximum degree in the actually considered graph. 

The following technical lemma confirms 
that \textsc{heavy stable sets} (Figure~\ref{figure procedure heavy}) terminates. 
Moreover, Lemma~\ref{theorem lemma procedure analysis} lists some properties of the vertex sets which were removed from the input graph. 

\begin{lemma}\label{theorem lemma procedure analysis}
 Let $G$ be a graph, and let $k \in \{2,\ldots,\chi(G)-1\}$. 
 If $G$ and $k$ are taken as input for \textsc{heavy stable sets},  
 then the algorithm terminates. 
 Moreover, for all $r \in \{0,\ldots,2k-1\}$,  
 \begin{enumerate}
    \item $\Delta(G_{r+1})\leq \Delta(G)-(r+1)$, with equality if and only if $S_{r}\not = \emptyset$,\label{lemma Delta in r}
    \item $G_{r+1} = G-\bigcup_{i=0}^{r}S_i$, \label{lemma Gr+1}
    \item $S_r$ is a (possibly empty) stable set, \label{lemma Sr empty set}
    \item for all $v \in S_r$, $v$ has $\Delta(G)-r$ neighbors in $G-\bigcup_{i=0}^rS_i$. \label{lemma neighbors in remainder}
  \end{enumerate}
\end{lemma}

\begin{proof}
 The algorithm starts with $r=0$.
 Since $G=G_0-S_0$, $\Delta(G_0-S_0) \geq \Delta(G)-0$ is trivially true. 
 Thus, the algorithm enters the outer while-loop.
 As long as 
 \begin{equation}\label{outer while condition}
  \Delta(G_r-S_r) = \Delta(G)-r 
 \end{equation}
 holds for any fixed $r \in \{0,\ldots,2k-1\}$, 
 the outer while-loop is repeated. 
 After every enlargement of $S_r$, 
 the procedure tests if~(\ref{outer while condition}) is still true 
 and enters the inner while-loop if and only if it is not. 
 The inner while-loop is repeated 
 until $r$ is large enough such that~(\ref{outer while condition}) 
 is true again or until $r$ reaches the value of $2k-1$, 
 forcing the algorithm to stop. 
 In particular, every time the inner while-loop is entered or iterated, 
 either the algorithm stops or $r$ is increased by one. 
 So, in order to show that the procedure terminates, 
 it suffices to prove that for every $r \in \{0,\ldots,2k-1\}$, 
 the inner while-loop is entered or iterated. 
 To attain this result, 
 let $S_r^i$ be the set $S_r$ after the $i$th vertex was added, 
 for some $i \in \mathbb{N}_0$. 
 Since $|S_r^i|=|i|$, 
 $$G_r-S_r^{i+1}\subsetneq G_r-S_r^{i}$$
 holds.  
 That is, every time the outer while-loop is repeated for some fixed $r$, 
 the considered graph is a proper subgraph of the graph considered in the previous iteration.
 Therefore, at some point of the iteration, 
 \begin{equation}\label{inner while condition}
  \Delta(G_r-S_r) < \Delta(G)-r 
 \end{equation}
	holds, contradicting (\ref{outer while condition}). 
 Hence, the inner while-loop is entered and the algorithm terminates. 

 Point 1 of Lemma~\ref{theorem lemma procedure analysis}: Observe that~(\ref{inner while condition}) implies 
 $\Delta(G_{r+1})=\Delta(G_r-S_r)\leq \Delta(G)-r-1$
 and that $S_r=\emptyset$ if and only if~(\ref{outer while condition}) 
 was false at any time of the procedure.  

Point 2 of Lemma~\ref{theorem lemma procedure analysis}: 
 Let $r \in \{0,\ldots,2k-1\}$.
 Since $G_{r+1}=G_{r}-S_{r}$ 
 and $G_0=G$, 
 $$G_{r+1} = G_r-S_r = G_{r-1}-S_{r-1}-S_r = \ldots = G_0-\bigcup_{i=0}^{r}S_i=G-\bigcup_{i=0}^{r}S_i.$$
 
Point 3 and 4 of Lemma~\ref{theorem lemma procedure analysis}: Let $S_r \not = \emptyset$. 
 Note that a vertex $v$ is put into $S_r$ after the conditions 
 $\Delta(G_r-S_r)=\Delta(G)-r$ and 
 $deg_{G_r-S_r}(v)=\Delta(G)-r$
 are verified. 
 In particular, 
 $v$ has $\Delta(G)-r$ neighbors in $$G-\bigcup_{i=0}^{r-1}S_i,$$
 and Point~\ref{lemma neighbors in remainder} follows if $v$ has no neighbors in $S_r$.  
 After $v$ is put into $S_r$, 
 the former neighbors of $v$ in the updated graph $G_r-S_r$ 
 have $(G_r-S_r)$-degree at most $\Delta(G_r-S_r)-1$.
 Hence, they will not be put into $S_r$. 
Hence, $S_r$ is stable and 
thus Point~\ref{lemma Sr empty set} and Point~\ref{lemma neighbors in remainder} follow. 
\end{proof}

We apply \textsc{heavy stable sets} to a minimal counterexample to Reed's Conjecture. 
Roughly speaking, 
the union of the constructed stable sets 
induces a color-critical subgraph in the counterexample.  
Every vertex in this union has relatively high degree. 

\begin{theorem} \label{theorem heavys allg}
Let $G$ be a minimal counterexample to Reed's Conjecture. 
Then $G$ contains a $2$-color-critical subgraph consisting 
of vertices with $G$-degree $\Delta(G)$. 
Moreover, for every $k \in \mathbb{N}$, $2 \leq k \leq \chi(G)-1$, 
$G$ contains a $(k+1)$-color-critical subgraph $H$ such that 
$deg_G(v) \geq \max \{\delta(G), \Delta(G)-2k+3\}$ holds for 
all $v \in V(H)$.
\end{theorem}

\begin{proof}
Let $k \in \{1,\ldots,\chi(G)-1\}$ and let $G$ be a minimal counterexample. 
Use $G$ and $k$ as input for \textsc{heavy stable sets}. 
Let $S$ be the graph induced by $$\bigcup_{i=0}^{2k-1}S_i$$ 
and assume that $S$ is $k$-colorable.
By Lemma~\ref{theorem lemma procedure analysis}, Point~\ref{lemma Gr+1}, 
$$G_{2k}=G-\bigcup_{i=0}^{2k-1}S_i=G-V(S)$$ holds, thus  
$\chi(G) \leq \chi(G_{2k}) + \chi(S)$. 
That is, $\chi(G_{2k}) \geq \chi(G) - k.$
By Lemma~\ref{theorem lemma procedure analysis}, Point~\ref{lemma Delta in r}, 
$\Delta(G) - 2k \geq \Delta(G_{2k}).$
Since $G_{2k}$ is an induced subgraph of $G$, 
$\omega(G) \geq \omega(G_{2k})$ holds. 
Given that $G$ is a counterexample, 
\begin{eqnarray*}
\chi(G_{2k}) 	&\geq& \chi(G)-k \\
		&>& \left \lceil \frac{\Delta(G)+\omega(G)+1}{2} \right \rceil - k\\
		&=& \left \lceil \frac{\Delta(G)-2k+\omega(G)+1}{2} \right \rceil \\
		&\geq& \left \lceil \frac{\Delta(G_{2k})+\omega(G_{2k})+1}{2} \right \rceil\\
\end{eqnarray*}
holds. 
Note that $k > 0$ and $S_0 \not = \emptyset$. 
Hence, $G_{2k}$ is both a proper induced subgraph of $G$ and 
a counterexample, contradicting the choice of $G$.
Thus, $S$ is at least $(k+1)$-chromatic. 
Due to Proposition~\ref{theorem proposition k vertex critical subgraphs}, 
$S$ contains an induced $(k+1)$-color-critical subgraph.

Let $H$ be such a $(k+1)$-color-critical subgraph and let $v \in V(H)$. 
If $k=1$, 
then $v \in S_0$ or $v\in S_1$. 
If $v\in S_0$, then $deg_G(v)=\Delta(G)$.
If $v\in S_1$, then $v$ has at least one neighbor in 
$H \subseteq (S_0 \cup S_1)$, 
by Proposition~\ref{theorem proposition min degree vertex critical}.  
Due to Lemma~\ref{theorem lemma procedure analysis}, Point~\ref{lemma neighbors in remainder}, 
$v$ has $\Delta(G)-1$ neighbors in $G-(S_0 \cup S_1)$, 
hence $deg_G(v)=\Delta(G)$. 
If $k\geq 2$, 
it suffices to prove that 
$deg_G(v) \geq \Delta(G)-2k+3$. 
To obtain this result, let $v \in V(H)$. 
Then $v \in S_j$ for some $j \in \{0,\ldots,2k-1\}$. 
Note that if $j \leq 2k-3$, then, 
by Lemma~\ref{theorem lemma procedure analysis}, Point~\ref{lemma neighbors in remainder}, 
$v$ has $\Delta(G)-j$ neighbors in $G-\bigcup_{i=0}^jS_i$. 
In this case, 
the claim follows, 
since $$deg_G(v) \geq \Delta(G)-j \geq \Delta(G)-2k+3.$$
Further note that if $j=2k-1$, then, 
again by Lemma~\ref{theorem lemma procedure analysis}, Point~\ref{lemma neighbors in remainder}, 
$v$ has $\Delta(G)-2k+1$ neighbors in $G_{2k}=G-S$. 
Moreover, $H \subseteq S$, 
and by 
\mbox{Proposition~\ref{theorem proposition min degree vertex critical}}, 
$v$ has at least $k$ neighbors in $H$.
It follows that
\begin{equation}\label{equation heavy S2k-1}
deg_G(v)\geq \Delta(G)-2k+1 + k = \Delta(G)-k+1 \geq \Delta(G)-2k+3.
\end{equation}
Hence in order to complete the proof, 
we have to consider the case $j=2k-2$.
Recall that by Lemma~\ref{theorem lemma procedure analysis}, Point~\ref{lemma Gr+1}, 
$$G_{2k-1}=G-\bigcup_{i=0}^{2k-2}S_i$$
holds. 
By Lemma~\ref{theorem lemma procedure analysis}, Point~\ref{lemma neighbors in remainder}, 
$v$ has $\Delta(G)-2k+2$ neighbors in $G_{2k-1}$.
If $v$ has at least one neighbor in $$\bigcup_{i=0}^{2k-2}S_i,$$ 
the proof is completed. 
Assume the contrary. 
Since $H$ is a $(k+1)$-color-critical graph, 
$v$ has at least $k$ neighbors in $H$.
By our assumption, 
every neighbor of $v$ in $H$ was added to the vertex set that induces $H$ \textit{after} $v$ was added. 
Hence, 
\begin{equation}\label{equation heavy degree v}
deg_G(v)=\Delta(G)-2k+2   
\end{equation}
and $N_H(v) \subseteq S_{2k+1}.$  
Let $x \in N_H(v)$. 
According to~(\ref{equation heavy S2k-1}), 
\begin{equation}\label{equation heavy degree x}
deg_G(x)\geq \Delta(G)-2k+3 
\end{equation}
holds. Recall that by Lemma~\ref{theorem lemma procedure analysis}, Point~\ref{lemma Gr+1}, 
$G_{2k-2}=G-\bigcup_{i=0}^{2k-3}S_i$ 
holds. 
Since $v \in S_{2k-2}$, 
by Lemma~\ref{theorem lemma procedure analysis},~Point~\ref{lemma Delta in r}, 
we have $\Delta(G_{2k-2})=\Delta(G)-2k+2.$
Therefore, the $G_{2k-2}$-degree of $x$ is not only at most $\Delta(G)-2k+2$ 
but is exactly of that value. 
Otherwise, since $v$ and $x$ are adjacent, 
the $G_{2k-1}$-degree of $x$ is at most $\Delta(G)-2k$, 
contradicting the fact that $x \in S_{2k-1}$. 
Thus, 
the $G_{2k-1}$-degree of $x$ is identical with the 
$G_{2k-1}$-degree of $v$. 
Choosing $v$ and not $x$ to be put into $S_{2k-2}$ contradicts the tie-break rule of the algorithm,   
since the $G$-degree of $x$ 
is higher than the $G$-degree than $v$, 
which can be seen by~(\ref{equation heavy degree x}) 
and~(\ref{equation heavy degree v}). 
This completes the proof. 
\end{proof}

Proposition~\ref{theorem proposition Toft Koenig1950} 
provides the sets of 2-~and 3-critical graphs. 
Hence, 
we adapt the result of Theorem~\ref{theorem heavys allg}.  

\begin{corollary}\label{theorem corollary heavys}
Let $G$ be a minimal counterexample to Reed's Conjecture. 
Then $G$ contains a heavy edge and a heavy odd cycle.
\end{corollary}

\begin{proof}
Let $G$ be a vertex-minimal counterexample to Conjecture~1. 
By Theorem~\ref{theorem heavys allg}, 
$G$ contains a $2$- respectively $3$-color-critical subgraph $H$ such that for all $v \in V(H)$, 
$deg_G(v)=\Delta(G)$ respectively $deg_G(v)\geq\Delta(G)-1$. 
By Proposition~\ref{theorem proposition Toft Koenig1950}, $H$ is an edge respectively $H$ is an odd cycle.
\end{proof}

Corollary~\ref{theorem corollary heavys} allows us to find two classes of graphs for which Reed's Conjecture holds. 

\begin{theorem}
\label{theorem fulfill Delta0}
Let $\Delta_0 \in \mathbb{N}$. 
If Reed's Conjecture holds for graphs with maximum degree at most $\Delta_0$, 
then the conjecture holds 
\begin{enumerate}
 \item for graphs in which the vertices of degree at least $\Delta_0 + 1$ form a stable set.\label{lemma stable set delta 0} 
 \item for graphs in which every induced odd cycle contains a vertex of degree at most $\Delta_0-1$. \label{lemma cycle delta 0}
\end{enumerate}
In particular, if for a graph $G$ 
the vertices of degree at least 5 form a stable set or 
if all induced cycles in $G$ that are of odd length contain a vertex of degree at most 3,  
then Reed's Conjecture holds for $G$.
\end{theorem}

\begin{proof}
Let the conjecture hold for graphs with maximum degree at most $\Delta_0$. 
Then every counterexample contains a minimal counterexample that has maximum degree at least $\Delta_0 + 1$ as induced subgraph.  
According to \mbox{Corollary~\ref{theorem corollary heavys}}, 
every minimal counterexample contains a heavy edge (respectively a heavy cycle). 
Hence, every counterexample contains  
an edge where the endvertices are of degree at least $\Delta_0+1$ 
(respectively an odd cycle where all vertices in the cycle are of degree at least $\Delta_0$). 
Thus, 
Claim~\ref{lemma stable set delta 0} (respectively Claim~\ref{lemma cycle delta 0}) follows. 

In particular, 
by Proposition~\ref{theorem proposition CE has Delta at least 5}, 
Reed's Conjecture holds for all graphs with $\Delta=4$. 
\end{proof}

If all induced odd cycles contain a vertex of degree at most 3, 
then there is another, straightforward way to validate Reed's Conjecture. 
A counterexample $G$ is at least $5$-colorable, 
since the conjecture holds for graphs with $\chi \leq 4$.
Thus, 
$G$ contains a $5$-critical graph 
whose minimum degree is at least $4$. 
But every subgraph of $G$ which is induced by vertices of degree 
at least $4$ is bipartite, a contradiction. 
However, color-critical graphs are not excluded by means of the minimum degree 
if the difference between the maximum degree and the clique number is large enough.

\begin{observation}
Let $\Delta_0 \in \mathbb{N}$. 
Graphs in which the vertices of degree at least $\Delta_0+1$ form a stable set 
and graphs in which every induced odd cycle contains a vertex of degree at 
most $\Delta_0-1$ can be recognized in polynomial time. 
\end{observation}

The first class mentioned in the observation 
is recognized by simply checking if an edge is formed 
in the set of vertices with degree at least $\Delta_0+1$. 
The recognition of the second class 
requires, essentially, testing if the graph induced by 
vertices of degree at least $\Delta_0+1$ is bipartite. 

\begin{theorem}\label{theorem coloring}
Let $G$ be a graph in which the 
the vertices of degree at least $5$ form a stable set or 
in which all induced odd cycles contain a vertex of degree at most 3.   
Then a coloring of $G$ that obeys Reed's Conjecture can be found in polynomial time. 
\end{theorem}
\begin{proof}
Let $G$ be a graph 
in which all induced odd cycles contain a vertex of degree at most 3. 
Let $B'$ be the graph induced by all vertices of degree at least $4$. 
Observe that $B'$ is bipartite. 
Let $B$ be an inclusionwise maximal bipartite subgraph of $G$ that contains $B'$. 
We use breadth-first search in order to color $B'$ with the colors $1$ and~$2$. 
Let $R$ be the graph induced by $V \setminus V(B)$. 
Assume $R$ contains a vertex of degree at least $2$, say $v$. 
Then the $R$-degree of $v$ is at least $4$, since otherwise, $v$ can be added to $B$. 
This contradicts the choice of $B$. 
Hence, the maximum degree of $R$ is $1$. 
Thus, $R$ is bipartite. 
We therefore again use breadth-first search 
to color $R$ with colors $3$ and $4$.  
This way, we provide a $4$-coloring. 

Let $G$ be a graph in which all vertices of degree at least $5$ form a stable set. 
Let $S$ be an inclusionwise maximal stable set which contains all vertices of degree at least $5$. 
Observe that $R$, 
which is the graph induced by $V(G) \setminus S$, 
has $\Delta(R) \leq 3$: 
Every hypothetical vertex with $R$-degree $4$ has $G$-degree at least $5$ and is therefore in $S$. 
We color $S$ with one color. 
A $3$-coloring of $R$ is provided by Lov\'{a}sz~\cite{Lovasz_1975}.     
\end{proof}

\bibliographystyle{plain}
\bibliography{./D_mybib_heavy}

\end{document}